\documentclass[11pt]{article}
\usepackage{amsmath, amsthm,amssymb,mathtools}  
\usepackage{fullpage}
\usepackage{authblk}
\usepackage{amssymb}
\usepackage{bold-extra}
\usepackage{xcolor}
\usepackage{titlesec}
\titlespacing*{\paragraph}{0pt}{1ex plus 1ex minus .2ex}{1em}
\usepackage{graphicx}
\usepackage{epstopdf}
\usepackage{seqsplit}
\usepackage{url}
\usepackage[normalem]{ulem}
\usepackage{enumerate}
\usepackage{hyperref}

\newtheorem{theorem}{Theorem}[section]
\newtheorem{lemma}[theorem]{Lemma} 
\newtheorem{proposition}[theorem]{Proposition} 
\newtheorem{corollary}[theorem]{Corollary}

\newcommand{\cE}{{\mathcal E}}
\newcommand{\cL}{{\mathcal L}}

\newcommand{\cS}{{\mathcal S}}
\newcommand{\cT}{{\mathcal T}}
\newcommand{\cN}{{\mathcal N}}

\title{On the Complexity of Optimising Variants of Phylogenetic Diversity on Phylogenetic Networks}
\author[1]{Magnus Bordewich}
\author[2]{Charles Semple}
\author[3]{Kristina Wicke}
\affil[1]{Department of Computer Science, Durham University, United Kingdom}
\affil[2]{School of Mathematics and Statistics, University of Canterbury, Christchurch, New Zealand}
\affil[3]{Department of Mathematics, The Ohio State University, Columbus, OH, USA}
\date{}                                           % Activate to display a given date or no date

\begin{document}
\maketitle

\begin{abstract}
Phylogenetic Diversity (PD) is a prominent quantitative measure of the biodiversity of a collection of present-day species (taxa). This measure is based on the evolutionary distance among the species in the collection. Loosely speaking, if $\cT$ is a {rooted} phylogenetic tree whose leaf set $X$ represents a set of species and whose edges have real-valued lengths (weights), then the PD score of a subset $S$ of $X$ is the sum of the weights of the edges of the minimal subtree of $\cT$ connecting the species in $S$. In this paper, we define several natural variants of the PD score for a subset of taxa which are related by a known {rooted} phylogenetic network. Under these variants, we explore, for a positive integer $k$, the computational complexity of determining the maximum PD score over all subsets of taxa of size $k$ when the input is restricted to different classes of {rooted} phylogenetic networks.
\end{abstract}


\section{Introduction}

Phylogenetic diversity (PD) is a popular measure for quantifying the biodiversity of a set of species based on their evolutionary history and relatedness. Roughly speaking, the PD score of a group of species (taxa) quantifies how much of the `tree of life' is spanned by the species in the group. Ever since its introduction by Daniel Faith in 1992~\cite{Faith1992}, this metric has attracted great attention in the literature both among empiricists and theorists. Indeed, Faith's {seminal} paper has been cited in excess of 2000~times.

In the face of limited resources in biodiversity conservation, a central problem in relation to phylogenetic diversity is to identify subsets of species that maximise the PD score. While there are efficient algorithms for finding maximum PD sets on a given tree \cite{Minh2006, Pardi2005, Steel2005}, variants of the problem (e.g., maximising PD across several trees \cite{bor09a, Spillner2008} or incorporating conservation costs and extinction probabilities in the so-called `Noah's ark problem' \cite{Weitzman1998}) have led to many interesting algorithmic questions. Most of this work to date has focused on measuring and maximising PD on phylogenetic trees. However, the metric has been extended to and analysed for so-called split networks \cite{bor12, Chernomor2016, min09, Moulton2007, Spillner2008} which are {typically} used to represent conflicts in data. More recently, the authors in~\cite{WICKE201880} have suggested approaches to measuring PD on explicit phylogenetic networks, which represent the evolutionary histories of collections of species whose past include reticulation (non-treelike) events such as hybridisation and horizontal gene transfer.

As processes such as hybridisation pose new challenges to biodiversity conservation (e.g.,~\cite{PachecoSierra2018,Quilodrn2020}) and diversity measures beyond PD on phylogenetic trees are needed, in this paper we present the first rigorous analysis of the computational complexity of optimising variants of phylogenetic diversity extended to rooted phylogenetic networks. These results could lead to algorithms that aid conservationists and policy makers in making more accurately informed decisions. We extend the work of \cite{WICKE201880} and define four natural variants of the PD score for a subset of taxa whose evolution is described by a rooted phylogenetic network $\cN$. We explore the relationships between these four measures and then analyse the computational complexity of, given $\cN$ and a positive integer $k$, determining the maximum PD score over all subsets of taxa of size $k$ under these variants of PD and for different classes of {rooted} phylogenetic networks.

The main results of this paper are as follows. We show that the complexity of determining the maximum AllPaths-PD score, our first variant of PD for rooted phylogenetic networks, depends on the class of networks to which $\cN$ belongs. In particular, for tree-child networks the optimisation problem is hard, whereas for level-1 networks the problem is polynomial (see Section~\ref{sec:AllPaths}). In Section~\ref{sec:network-PD} we introduce a second variant, Network-PD, which, in contrast to AllPaths-PD, takes into account the proportion of features a reticulation vertex inherits from each of its parents. We show that Network-PD is a generalisation of AllPaths-PD, and thus the corresponding optimisation problem is again computationally hard in general. In addition, we show that there is a direct correspondence between the maximum and minimum of Network-PD, and the third and fourth variants of PD for phylogenetic networks considered in this manuscript, namely MaxWeightTree-PD and MinWeightTree-PD. We end the paper by analysing the latter two more in-depth. More precisely, we show that the problem of determining the maximum value over all subsets of taxa of size $k$ is solvable in polynomial time for MaxWeightTree-PD (Section~\ref{sec:max}), whereas for MinWeightTree-PD even computing the MinWeightTree-PD score of a fixed subset of $X$ is computationally hard, and hence the optimisation problem is also hard (Section~\ref{sec:min}).

Before the main results, in Section~\ref{sec:notation} we give formal definitions of the structures and notation used throughout this manuscript. After reviewing the concept of PD on {rooted} phylogenetic trees, we then formally introduce our four variants of PD on {rooted} phylogenetic networks in Section~\ref{sec:PDvariants}. As described above, the remaining sections are devoted to analysing the complexity of determining the maximum PD score over all subsets of taxa of size $k$ under these four variants of PD and when the input is restricted to different classes of {rooted} phylogenetic networks.

\section{Notation and preliminaries}
\label{sec:notation}

To formally state our results, we need some notation and terminology. Throughout the paper, $X$ denotes a non-empty finite set (of taxa). 

\paragraph{Phylogenetic networks.}
A \emph{rooted binary phylogenetic network $\cN$ on $X$} is a rooted directed acyclic graph with no parallel arcs satisfying the following properties:
\begin{enumerate}[(i)]
    \item the (unique) root has in-degree zero and out-degree two;
    \item a vertex with out-degree zero has in-degree one, and the set of vertices with out-degree zero is $X$; {and}
    \item all other vertices have either in-degree one and out-degree two, or in-degree two and out-degree one.
\end{enumerate}
For technical reasons, if $|X|=1$, we allow $\cN$ to consist of the single vertex in $X$. The vertices in $X$ are \emph{leaves}. We call $X$ the {\it leaf set} of $\cN$ and frequently denote it by $\cL(\cN)$. The vertices with in-degree one and out-degree two are \emph{tree vertices}, while the vertices with in-degree two and out-degree one are \emph{reticulations}. We refer to arcs directed into a reticulation as \emph{reticulation arcs} and to all other arcs as \emph{tree arcs}. Furthermore, throughout the paper, we assume that all arcs of $\cN$ have non-negative real-valued lengths, that is, if $A$ denotes the arc set of $\cN$, then associated with $\cN$ is a mapping $w: A\rightarrow \mathbb R^{\ge 0}$ under which each arc $e$ of $\cN$ is assigned the weight $w(e)$. To illustrate, three rooted binary phylogenetic networks are shown in Fig.~\ref{Fig_NetworkTypes}. Here, as in all figures in the paper, arcs are directed down the page. A \emph{rooted binary phylogenetic $X$-tree} is a rooted binary phylogenetic network on $X$ with no reticulations. For the remainder of the paper, we will refer to rooted binary phylogenetic networks and rooted binary phylogenetic trees as phylogenetic networks and {phylogenetic} trees, respectively, as all such networks and trees considered are rooted and binary.

\paragraph{Tree-child and level-$1$ networks.}
A phylogenetic network $\cN$ on $X$ is a \emph{tree-child network} \cite{Cardona2009} if each non-leaf vertex is the parent of a tree vertex or a leaf. Equivalently, $\cN$ is tree-child if (i) no tree vertex is the parent of two reticulations and (ii) no reticulation is the parent of another reticulation~\cite{Semple2015}. And again, equivalently, $\cN$ is tree-child if, for every vertex $v$ of $\cN$, there is a path from $v$ to a leaf $\ell$ that consists only of tree vertices (except $\ell$ and possibly $v$ itself). We call such a path a \emph{tree path}.

Let $\cN$ be a phylogenetic network. A reticulation arc $(u,v)$ of $\cN$ is called a \emph{shortcut} if there is a directed path in $\cN$ from $u$ to $v$ that avoids $(u,v)$. We say that $\cN$ is a \emph{normal network} \cite{Willson2009} if it is tree-child and has no shortcuts. Finally, $\cN$ is a \emph{level-$1$ network} if its underlying (undirected) cycles are vertex disjoint. Normal and level-$1$ networks are {proper} subclasses of tree-child networks. An example of two tree-child networks, one normal and the other level-$1$, is shown in Fig.~\ref{Fig_NetworkTypes}.

\begin{figure}[htbp]
    \centering
    \includegraphics[scale=0.2]{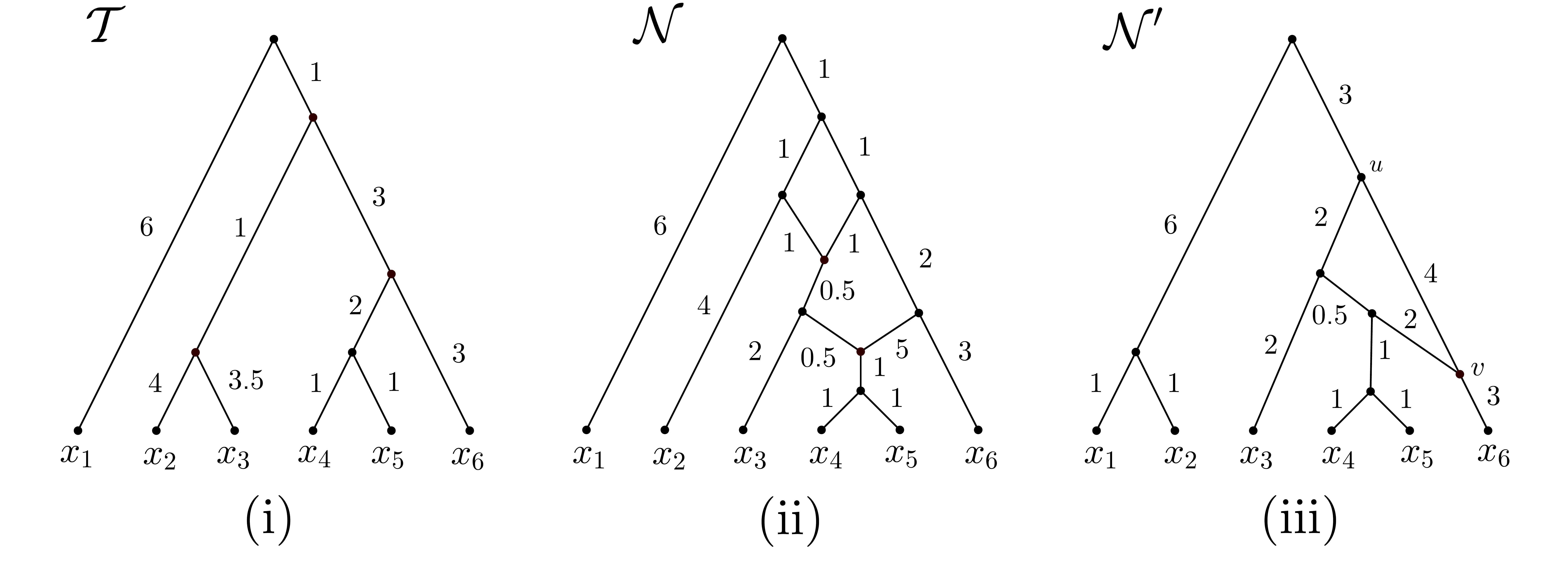}
    \caption{Phylogenetic tree $\cT$ and tree-child networks $\cN$ and $\cN'$ on $X=\{x_1, \ldots, x_6\}$ with non-negative {real-valued} arc weights. Note that $\cN$ is normal but not level-1, whereas $\cN'$ is level-1 but not normal (due to the shortcut $(u,v)$).}
    \label{Fig_NetworkTypes}
\end{figure}

\paragraph{Connecting subtrees.} Let $\cN$ be a phylogenetic network on $X$ with root $\rho$, and let $Y \subseteq X$. We call any subgraph  $T$ of $\cN$ that is a directed rooted tree (i.e. an arborescence) with root $\rho$ and leaf set $Y$, a \emph{connecting subtree for $Y$}. Note that $\rho$ may have out-degree one in $T$. Moreover, there might be several connecting subtrees for $Y$ in $\cN$. {W}e denote the set of all connecting subtrees for $Y$ in $\cN$ as $\mathsf{T}_Y(\cN)$. Also note that any connecting subtree $T$ for $Y$ is an edge-weighted tree, where each edge $e \in T$ inherits its weight $w(e)$ from $\cN$.

\section{Variants of phylogenetic diversity for phylogenetic networks} \label{sec:PDvariants}

In this section, we introduce our variants of PD for phylogenetic networks and then consider the associated optimisation problems.

\subsection{PD on {phylogenetic} trees and {phylogenetic} networks}

Before we define variants of PD for phylogenetic networks, we briefly review PD for {phylogenetic} trees. Phylogenetic diversity arose as a quantitative measure of the biodiversity of a set of species for use in conservation decisions~\cite{Faith1992}. PD has been studied for a variety for organisms, ranging from bacteria~\cite{Lozupone2007}, to plants~\cite{Cadotte2008}, to mammals~\cite{Safi2011}. Moreover, the \emph{International Union for Conservation of Nature} (IUCN) has established a `Phylogenetic Diversity Task Force' aiming at promoting the use of PD in conservation decisions (see \url{https://www.pdtf.org/}). PD also serves as a basis for so-called phylogenetic diversity indices such as the `Fair Proportion index'~\cite{Redding2003} (also called `evolutionary distinctiveness score'~\cite{Isaac2007}) {and} the `Equal Splits index'~\cite{Redding2003,Redding2006} that rank species for conservation, based on their contribution to overall PD. These indices are used in conservation initiatives such as the `EDGE of Existence programme' established by the Zoological Society of London~\cite{Isaac2007}.

The key underlying assumption in the use of PD as a quantitative measure is that if the arcs of a phylogenetic tree are weighted according to genetic distance, then features of interest (be that biological, pharmaceutical or conservational) will have arisen at a rate proportional to the lengths of the {arcs}. A further assumption is that all features that arose in an ancestral species have persisted to be present in the extant species descended from them. {So} the PD score is proportional to the number of distinct features present in a set of species. In particular, let $\cT$ be a phylogenetic $X$-tree (with non-negative real-valued {arc} lengths). The \emph{phylogenetic diversity} (PD score) of a subset $S \subseteq X$, denoted as $\textup{PD}_{\cT}(S)$, is the sum of {arc} lengths in the (unique) connecting subtree for $S$ in $\cT$.  Referring to Fig.~\ref{Fig_NetworkTypes}(i), if $S = \{x_2,x_4,x_6\}$, then $\textup{PD}_{\cT}(S) = 15$.

\paragraph{AllPaths-PD.}
There are different ways {that} the definition of PD may be extended from phylogenetic trees to phylogenetic networks, which we discuss now. The most straightforward approach is to again assume that all features that arise in any ancestral species persist to be present in all descendant extant species; then the natural extension of the PD score to networks is what we have called \emph{AllPaths-PD}. Specifically, for a phylogenetic network $\cN$ and a subset $S\subseteq X$, we define
$$\textrm{AllPaths-PD}_{\cN}(S) = \sum_{e\in \textrm{Anc}(S)}w(e),$$ 
where $\textrm{Anc}(S)$ is the set of all arcs that are ancestral to at least one taxon in $S$, i.e. lie on a directed path from the root of $\cN$ to some leaf in $S$.

\paragraph{Network-PD.}
To obtain a more accurate evaluation of the relative feature diversity of different subsets of taxa, we require knowledge of the proportion of features present in a parent species that are inherited by a child species. At a reticulation representing a true hybridisation event, the child taxon might inherit 50\% of the features of one parent and 50\% of the features of the other. Whereas, at a reticulation representing a lateral gene transfer, the child may inherit the entire genome of one parent species, i.e. 100\% of the features, and also receive an injection of a small section of DNA from the other parent, perhaps 5\% of the features. Thus, at each reticulation of our phylogenetic network, we must be given weights on each incoming arc corresponding to the proportion of features of the parent inherited by the child. Given this information, a more accurate measure, which we have called \emph{Network-PD}, may be obtained. 

On each incoming arc $e=(u,v)$ to each reticulation $v$ of a phylogenetic network $\cN$, let $p(e)$ be the \emph{inheritance proportion {(function})}, giving the proportion of features of the parent vertex $u$ that are present in the child vertex $v$ of that arc. We assume that for all reticulation arcs $p(e)\in[0,1]$. (Just as genetic distance is used in the {arc} lengths of a phylogenetic tree when computing PD as a proxy for the number of features of interest, we could use the proportion of the parental genome present in the child taxon as an estimate for the proportion of parental features inherited by the child.) For a subset $S$ of the leaves of $\cN$, {we define}, for each arc $e=(u,v)\in\cN$, {the function} $\gamma(S,e)$ to be the proportion of the features of $v$ that are present in the taxa set $S$. (Equivalently, $\gamma(S,e)$ is the probability that a feature arising on arc $e$ is inherited by some taxon in set $S$). We {now} define Network-PD as follows:
$$\textrm{Network-PD}_{\cN,\, p}(S) = \sum_{e\in \cN}\gamma(S,e)\cdot w(e).$$
Where the {phylogenetic} network $\cN$ or inheritance proportion function $p$ is obvious, we may omit them from the subscript. We may compute $\gamma(S,e)$ in a bottom up fashion as follows. For $e=(u,v)$, 
\begin{enumerate}[(i)]
    \item if $v$ is a leaf and $v\in S$ then $\gamma(S,e)=1$, whereas if $v\not\in S$ then $\gamma(S,e)=0$;
    \item if $v$ is a tree vertex with outgoing arcs $f_1$ {and} $f_2$, then $$\gamma(S,e)=\gamma(S,f_1)+\gamma(S,f_2)-\gamma(S,f_1)\gamma(S,f_2);$$
    \item if $v$ is a reticulation vertex with outgoing arc $f$, then $\gamma(S,e)=\gamma(S,f)p(e)$.
\end{enumerate}

\paragraph{MaxWeightTree-PD and MinWeightTree-PD.}
It is likely that, in practice, complete knowledge of the inheritance proportion function $p$ will not be possible, so we may be interested in upper and lower bounds on Network-PD under varying $p$. Note that if $p$ is allowed to vary without restriction it can still be no more than 1 on each arc, and that setting $p\equiv 1$ gives AllPaths-PD. The  inheritance proportion $p$ can also be no less than $0$ on any arc, and setting $p\equiv 0$ reduces Network-PD$_\cN$ to PD on a phylogenetic tree (specifically PD$_\cT$, where $\cT$ is, {up to isomorphism}, {the phylogenetic} tree obtained by deleting each reticulation arc of $\cN$ and connecting the reticulation vertices to the root by {arcs} of weight~$0$). These {extremities of $p$}
on Network-PD are thus not (mathematically) interesting in their own right. 

However, total inheritance proportions of 0 or 2 at a reticulation are unrealistic. Alternatively, we might assume that each feature inherited from the second parent replaces some feature from the first parent. That is to say, at a reticulation with incoming arcs $e_1$ and $e_2$ we require $p(e_1)+p(e_2)=1$. Under this assumption  upper and lower bounds for the total quantity of features present in a given subset of taxa are the PD scores of the maximum-weight and minimum-weight connecting subtrees for those taxa. Specifically, for a subset $S\subseteq X$, we define the following two variants of PD:
$$\textrm{MaxWeightTree-PD}_{\cN}(S) = \max_{T \in \mathsf{T}_S(\cN)} \sum_{e \in T}w(e)$$
and
$$\textrm{MinWeightTree-PD}_{\cN}(S) = \min_{T \in \mathsf{T}_S(\cN)} \sum_{e \in T}w(e).$$
We elaborate further how Network-PD is bounded by MinWeightTree-PD and MaxWeightTree-PD in Section~\ref{subsec:maxmin_network_pd}.

Note that AllPaths-PD is called `phylogenetic subnet diversity' in~\cite{WICKE201880}, MinWeightTree-PD is called `phylogenetic net diversity', and MaxWeightTree-PD is related to the notion of `embedded phylogenetic diversity' discussed therein. However, while the authors of~\cite{WICKE201880} introduced and compared different variants of PD for phylogenetic networks, they did not analyse the complexity of, given a phylogenetic network $\cN$ and a positive integer $k$, computing the maximum PD score over all subsets of taxa of size $k$, or finding a subset $S \subseteq X$ of cardinality $k$ which maximises the PD score under these variants. In the following, we consider the first problem, i.e. computing the maximum PD score over all subsets of taxa of size $k$ under the phylogenetic diversity variants introduced above on phylogenetic networks.

\subsection{Optimisation problems}

The problem of finding a subset of taxa of cardinality $k$ maximising the PD score has been extensively studied on phylogenetic trees. It corresponds to the task in conservation biology of determining which $k$ species maximise the biodiversity of the group~\cite{Faith1992}. Here, we focus on the problem of computing the maximum PD score over all subsets of taxa of size $k$ under the variants of PD defined above. More precisely, we study the following optimisation problems:

\noindent\textbf{\textsc{Max-AllPaths-PD}{$(\cN, k)$}}:\\ 
\textbf{Input:} A phylogenetic network $\cN$ on taxa set $X$ and a positive integer $k$.\\
\textbf{Objective:} Determine the maximum value of AllPaths-PD$_{\cN}(S)$ over all subsets $S\subseteq X$ of cardinality $k$.

\noindent\textbf{\textsc{Max-Network-PD}}{$(\cN, p, k)$}:\\ 
\textbf{Input:} A phylogenetic network $\cN$ on taxa set $X$, a inheritance proportion function $p$ on the reticulation arcs of $\cN$, and a positive integer $k$.\\
\textbf{Objective:} Determine the maximum value of Network-PD$_{\cN,\, p}(S)$ over all {subsets} $S\subseteq X$ of cardinality $k$.

\noindent\textbf{\textsc{Max-MaxWeightTree-PD}}{$(\cN, k)$}:\\ 
\textbf{Input:} A phylogenetic network $\cN$ on taxa set $X$ and a positive integer $k$.\\
\textbf{Objective:} Determine the maximum value of MaxWeightTree-PD$_{\cN}(S)$ over all subsets $S\subseteq X$ of cardinality $k$.

\noindent\textbf{\textsc{Max-MinWeightTree-PD}}{$(\cN, k)$}:\\ 
\textbf{Input:} A phylogenetic network $\cN$ on taxa set $X$ and a positive integer $k$.\\
\textbf{Objective:} Determine the maximum value of  MinWeightTree-PD$_{\cN}(S)$ over all subsets $S\subseteq X$ of cardinality $k$.

\noindent The complexity of these optimisation problems will be discussed in the sections that follow. 

\section{AllPaths-PD}
\label{sec:AllPaths}

We begin by studying AllPaths-PD.
Let $\cN$ be a phylogenetic network on $X$. Recall that, for any subset $S\subseteq X$, we defined AllPaths-PD$_{\cN}(S)$ to be the sum of the weights of all arcs of $\cN$ which lie on a path from the root of $\cN$ to a leaf in $S$, and \textsc{Max-AllPaths-PD} to be the problem of finding the maximum value of AllPaths-PD$_{\cN}(S)$ over all subsets $S\subseteq X$ of cardinality $k$.

In this section we show first that, in general, the problem \textsc{Max-AllPaths-PD} is NP-hard even when restricted to the class of normal networks. Moreover, we observe that
%under standard assumptions
\textsc{Max-AllPaths-PD} cannot be approximated within $1-\tfrac{1}{e}\approx 0.632$ {unless ${\rm P}={\rm NP}$}, and that a greedy algorithm will achieve this approximation ratio. Second, we show that \textsc{Max-AllPaths-PD} can be solved in polynomial time on level-1 networks.

\subsection{Maximising AllPaths-PD is NP-hard}

In order to show that Maximising AllPaths-PD is NP-hard we will use a reduction to the well-known Maximum Coverage problem:

\noindent\textbf{\textsc{Maximum Coverage}{$(\mathcal S, k)$}}:\\ 
\textbf{Input:} A collection $\mathcal{S}=\{S_1,S_2,...,S_n\}$ of sets and a positive integer $k$.\\
\textbf{Objective:} Find a subset $\mathcal{S'}\subseteq \mathcal{S}$ such that $|\mathcal{S'}|=k$ and the number of covered elements, $|\bigcup_{S_i\in\mathcal{S'}}S_i|$, is maximised.

\noindent \textsc{Maximum Coverage} is NP-hard to solve exactly. Indeed, the inapproximability of \textsc{Maximum Coverage} is well studied, and it is known that the approximation threshold is $1-\tfrac{1}{e}$. That is, unless ${\rm P}={\rm NP}$, no polynomial-time algorithm exists that {always} returns a solution to \textsc{Maximum Coverage} that is guaranteed to have value greater than $1-\tfrac{1}{e}$ of the optimal solution~\cite{fei98}.

\begin{theorem}\label{thm:np-hard}
The problem \textsc{Max-AllPaths-PD} is NP-hard. Moreover, \textsc{Max-AllPaths-PD} cannot be approximated in polynomial time with approximation ratio better than $1-\tfrac{1}{e}$ unless P$=$NP.
\end{theorem}

\begin{proof}
Let $\cS$, a collection of sets, and $k$, a positive integer, be an instance of \textsc{Maximum Coverage}. We begin by constructing a phylogenetic network $\cN$ with leaf set $\cS$ as follows. Set $E=\bigcup_{S\in\cS}$, that is, $E$ is the ground set of the \textsc{Maximum Coverage} instance. Take any phylogenetic tree (a caterpillar would do) with leaf set $E$, where each internal arc has weight~$0$ and each pendant arc (i.e. an arc incident with a leaf) has weight~$1$. The arcs of weight~$1$ are thus in one-to-one correspondence with {the elements in} $E$. Label each of these arcs with the same element of $E$ as its incident leaf. Now, for each $S\in \cS$, (i) add two new vertices $S$ {and} $\hat{S}$
to this construction {and} (ii), for each $e\in S$, add {a} new arc $(e, \hat S)$ and add {a} new arc $(\hat S, S)$.

This will result in each vertex labelled $e$, {where $e\in E$}, having out-degree corresponding to the number of sets containing $e$, and each vertex $\hat S$ having in-degree corresponding to the number of members of $S$. Next, refine every vertex that has either out-degree at least three or in-degree at least three, so that every resulting vertex has out-degree {equal to} two or in-degree {equal to} two, {respectively}. These new arcs below the arcs of $E$ are all assigned weight~$0$. Finally, suppress any vertices {of in-degree one and out-degree one} resulting from an element $e\in E$ that is only contained in a single member of $\cS$, and keep weight 1 and the label $e$ on the newly merged arc. The resulting phylogenetic network on $\cS$ is $\cN$ {and the construction takes time polynomial in the size of $\cS$}. See Fig.~\ref{fig:hardness} for an example of the construction.

\begin{figure}
\begin{center}
\includegraphics[width=\textwidth]{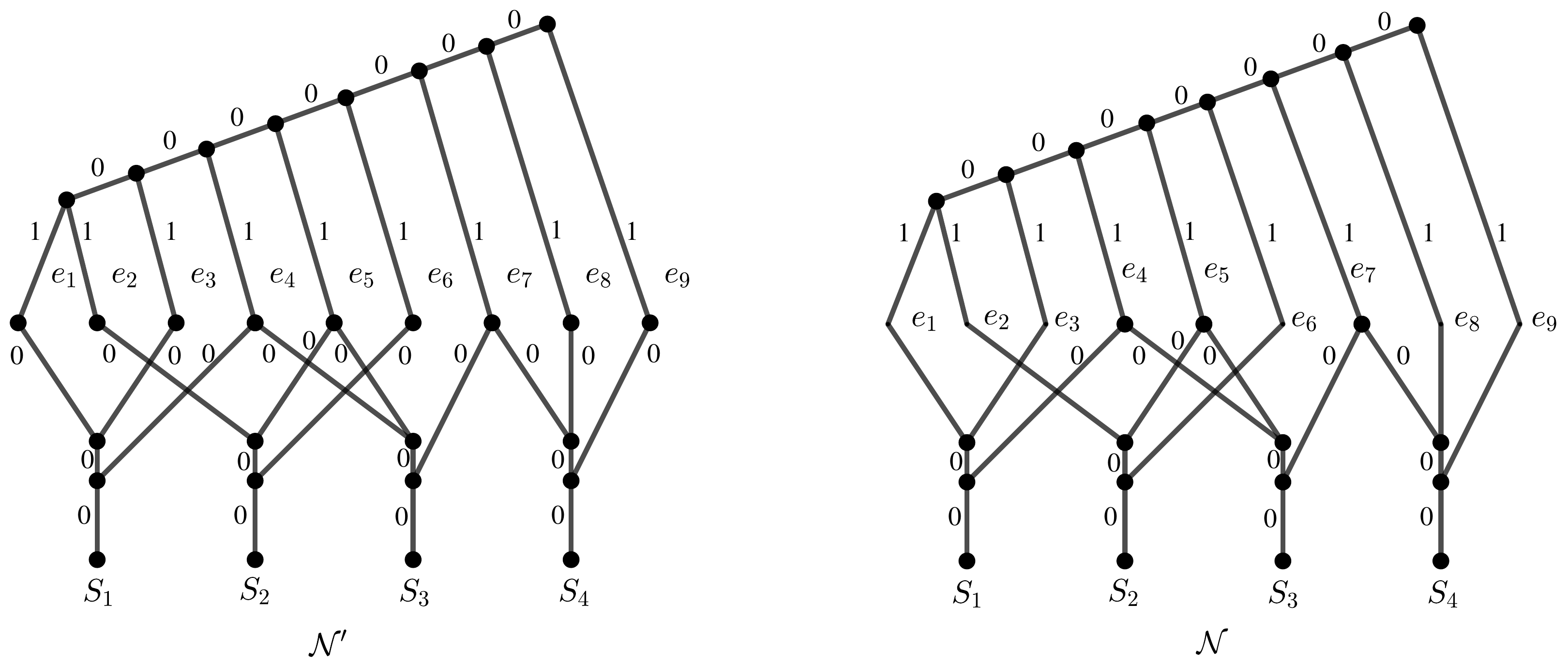}
\caption{The network $\cN$ resulting from reducing a \textsc{Maximum Coverage} instance $(\cS,k)$ to a \textsc{Max-AllPaths-PD} {instance} $(\cN, k)$. In this case, $\cS=\{S_1,S_2,S_3,S_4\}$, where $S_1=\{e_1,e_3,e_4\}, S_2=\{e_2,e_5,e_6\}, S_3=\{e_4,e_5,e_7\}$ and $S_4=\{e_7,e_8,e_9\}$, and so $E=\{e_1,e_2,\ldots,e_9\}$. On the left the network $\cN'$ illustrates the construction before the final step of suppressing {the} vertices {of in-degree one and out-degree one}.}
\label{fig:hardness} 
\end{center}
\end{figure}

Consider a subset $\cS'$ of the leaves of $\cN$. {An} arc $e\in E$ is on a path from the root {of $\cN$} to a member $S$ of $\cS'$ if and only if the set $S$ contains the corresponding element $e$. Thus the number of arcs in $E$ that lie on paths from the root to a member of $\cS$ is precisely $|\bigcup_{s\in\cS'}S|$. Since all the arcs of $\cN$ have weight~$0$ except those labeled with an element of $E$ which have weight~$1$, 
$$\textrm{AllPaths-PD}_{\cN}(\cS')=\left|\bigcup_{s\in\cS'}S\right|.$$
Thus solving \textsc{Max-AllPaths-PD} is at least as hard as solving \textsc{Maximum Coverage}, which is well known to be NP-hard~\cite{fei98}.

It is known that the approximation threshold {of} \textsc{Maximum Coverage} is $1-\tfrac{1}{e}$. Since we have equality in the optimal solutions to the two problems, if a polynomial-time algorithm that approximated \textsc{Max-AllPaths-PD} to a better ratio than $1-\tfrac{1}{e}$ existed, then using the reduction above we would be able to obtain a polynomial-time approximation for \textsc{Maximum Coverage} with the same ratio. Thus, unless P$=$NP, this is not possible. {This completes the proof of the theorem.}
\end{proof}

The hardness result of Theorem~\ref{thm:np-hard} can be extended to show that \textsc{Max-AllPaths-PD} remains NP-hard even when the input{ted} {phylogenetic} network is restricted to be from the class of normal networks.

\begin{theorem}\label{thm:normal}
The problem \textsc{Max-AllPaths-PD} is NP-hard even when the input{ted} {phylogenetic} network is restricted to be from the class of normal networks.
\end{theorem}

\begin{proof}
The construction is the same as used in the proof of Theorem~\ref{thm:np-hard} with the following additions. Starting from the constructed phylogenetic network $\cN$ with leaf set $\cS$, we transform it into a normal network $\cN'$ by 
\begin{enumerate}[(i)]
    \item assigning weight~$1$ to the pendant arcs leading to the leaves of $\cN$ labeled by elements of $\cS$ (which we call \emph{original} leaves); and
    \item subdividing all reticulation arcs of $\cN$ with a new vertex and adjoining a new leaf via a new pendant arc with weight~$0$ to each new vertex (which we call a \emph{new leaf}).
\end{enumerate}
Note that the construction in (ii) means that $\cN'$ is normal~\cite{Willson2009}. {Furthermore, this construction takes time polynomial in the size of $\cN$.} Now observe that, for any subset $\cS'$ of the leaf set of the augmented phylogenetic network $\cN'$, if $S'$ contains a new leaf $\ell$, then we may find an original leaf $\ell'$ that is a descendent of $\ell$'s parent vertex. The set $(\cS'\cup \{\ell'\})-\{\ell\}$ will have an AllPaths-PD score at least as high as that for the set $\cS'$, since all arcs with weight~$1$ on a path from the root of $\cN'$ to $\ell$ are also on a path from the root of $\cN'$ to $\ell'$. Moreover, if a subset $\cS''\subseteq \cS$ contains $k$ original leaves, then the AllPaths-PD score is precisely
$$k+\left|\bigcup_{s\in\cS''}S\right|.$$
Thus optimal solutions to \textsc{Max-AllPaths-PD} still correspond to optimal subsets of $\cS$, and thus \textsc{Max-AllPaths-PD} remains NP-hard even in this restricted case.
\end{proof}

We next show that AllPaths-PD is a submodular function. From this it will follow that the following greedy algorithm will yield a guaranteed approximation ratio of $1-\frac{1}{e}$ for \textsc{Max-AllPaths-PD}: first select a taxon at maximum distance from the root of $\cN$, and then iteratively select a taxon that maximises the incremental increase in AllPaths-PD until the requisite number of taxon have been selected.

\begin{lemma}
\label{lem:submodular}
Let $\cN$ be a phylogenetic network on $X$. Then the function \textup{AllPaths-PD}, which assigns each subset $S \subseteq X$ a non-negative real value, is a submodular function, i.e. for all $A, B  \subseteq X$, we have
$$ \textup{AllPaths-PD}_{\cN}(A \cup B) + \textup{AllPaths-PD}_{\cN}(A \cap B) \leq \textup{AllPaths-PD}_{\cN}(A) + \textup{AllPaths-PD}_{\cN}(B).$$
\end{lemma}

\begin{proof}
Recall that, for {all} $S\subseteq X$, $\textrm{Anc}(S)$ is the {set} of arcs that are ancestral to {at least one taxon} in $S$. Thus 
\begin{align*}
    \textrm{AllPaths-PD}(S) &= \sum\limits_{e \in \cN} \delta(S,e) \cdot w(e), 
\end{align*}
where
\begin{align*}
    \delta(S,e) &= \begin{cases}
    1, &\text{if } e \in \textrm{Anc}(S); \\
    0, &\text{otherwise.}
    \end{cases}
\end{align*}
If $e\in \textrm{Anc}(A\cup B)$, then $e$ is on a path from the root of $\cN$ to a leaf in $A\cup B$, that is, {$e$ is} on {a} path from the root of $\cN$ to a leaf in either $A$ or $B$. Thus
$$\textrm{Anc}(A \cup B) = \textrm{Anc}(A) \cup \textrm{Anc}(B)$$
and, similarly,
$$\textrm{Anc}(A \cap B) \subseteq \textrm{Anc}(A) \cap \textrm{Anc}(B).$$
It follows that, for all {arcs} $e$ {in} $\cN$ and all $A,B \subseteq X$, 
\begin{align*}
    \delta(A \cup B, e) + \delta(A \cap B, e) - \delta(A,e) - \delta(B,e) \leq 0.
\end{align*}
Since $$\textrm{AllPaths-PD}(A \cup B) + \textrm{AllPaths-PD}(A \cap B) - \textrm{AllPaths-PD}(A) - \textrm{AllPaths-PD}(B)$$ is a weighted sum of the corresponding $\delta$ quantities above with non-negative weights, the statement now follows.
\end{proof}

By Lemma~\ref{lem:submodular}, AllPaths-PD is a submodular function. It is also non-decreasing (adding a taxon will never decrease the AllPaths-PD score). Hence standard approaches to constructing greedy algorithms for cardinality constrained submodular functions~\cite{nem81} yield an approximation algorithm, and so we have the following immediate corollary. In particular, the greedy algorithm described prior to Lemma~\ref{lem:submodular} will give this approximation. Note that, by Theorem~\ref{thm:np-hard}, this approximation ratio cannot be improved unless ${\rm P}\neq {\rm NP}$.

\begin{corollary}
The greedy algorithm prior to Lemma~\ref{lem:submodular} returns a $1-\tfrac{1}{e}$ approximation for \textsc{Max-AllPaths-PD}. 
\end{corollary}

\subsection{Maximising AllPaths-PD on level-1 networks}

Although {\sc Max-AllPaths-PD} is NP-hard in general, in this section we show that it is polynomial time for the class of level-$1$ networks. Let $\cN$ be a level-$1$ network on $X$. We begin by determining two connecting subtrees $T_1$ and $T_2$ for $X$ in $\cN$ that together cover all arcs of $\cN$.
Let $v$ be a reticulation of $\cN$. Since $\cN$ is level-$1$, it is easily seen that there is a unique tree vertex, $s$ say, of $\cN$ such that there exists distinct (directed) paths $P$ and $P'$ starting at $s$ and ending at $v$ such that $s$ and $v$ are the only vertices of $P$ and $P'$ in common. We refer to $s$ as the {\em source} vertex of $v$. Now, let $v_1, v_2, \ldots, v_k$ denote the reticulations of $\cN$. For each $i\in \{1, 2, \ldots, k\}$, let $s_i$ denote the source vertex of $v_i$. Furthermore, for each $i\in \{1, 2, \ldots, k\}$, let $u_i$ and $u'_i$ denote the (distinct) parents of $v_i$. Note that at most one of $u_i$ and $u'_i$ is $s_i$. 
Let $T_1$ be the connecting subtree for $X$ in $\cN$ obtained from $\cN$ by deleting $(u'_i, v_i)$ for all $i$, and re-weighting each of the arcs on the (unique) path from $s_i$ to $u'_i$ as zero for all $i$. All other arcs of $T_1$ keep the weighting inherited from $\cN$. 
Similarly, let $T_2$ be the connecting subtree for $X$ in $\cN$ obtained from $\cN$ by deleting $(u_i, v_i)$ for all $i$, and re-weighting each arc not on the (unique) path from $s_i$ to $v_i$ via $u'_i$ as zero for all $i$.
We call $(T_1, T_2)$ a {\em weighted covering} of $\cN$. {To illustrate this construction}, an example is given in Fig.~\ref{Fig:Level1}, {where $\cN$ is a level-$1$ network, and $(T_1, T_2)$ is a weighted covering of $\cN$}.

\begin{figure}[htbp]
    \centering
    \includegraphics[scale=0.25]{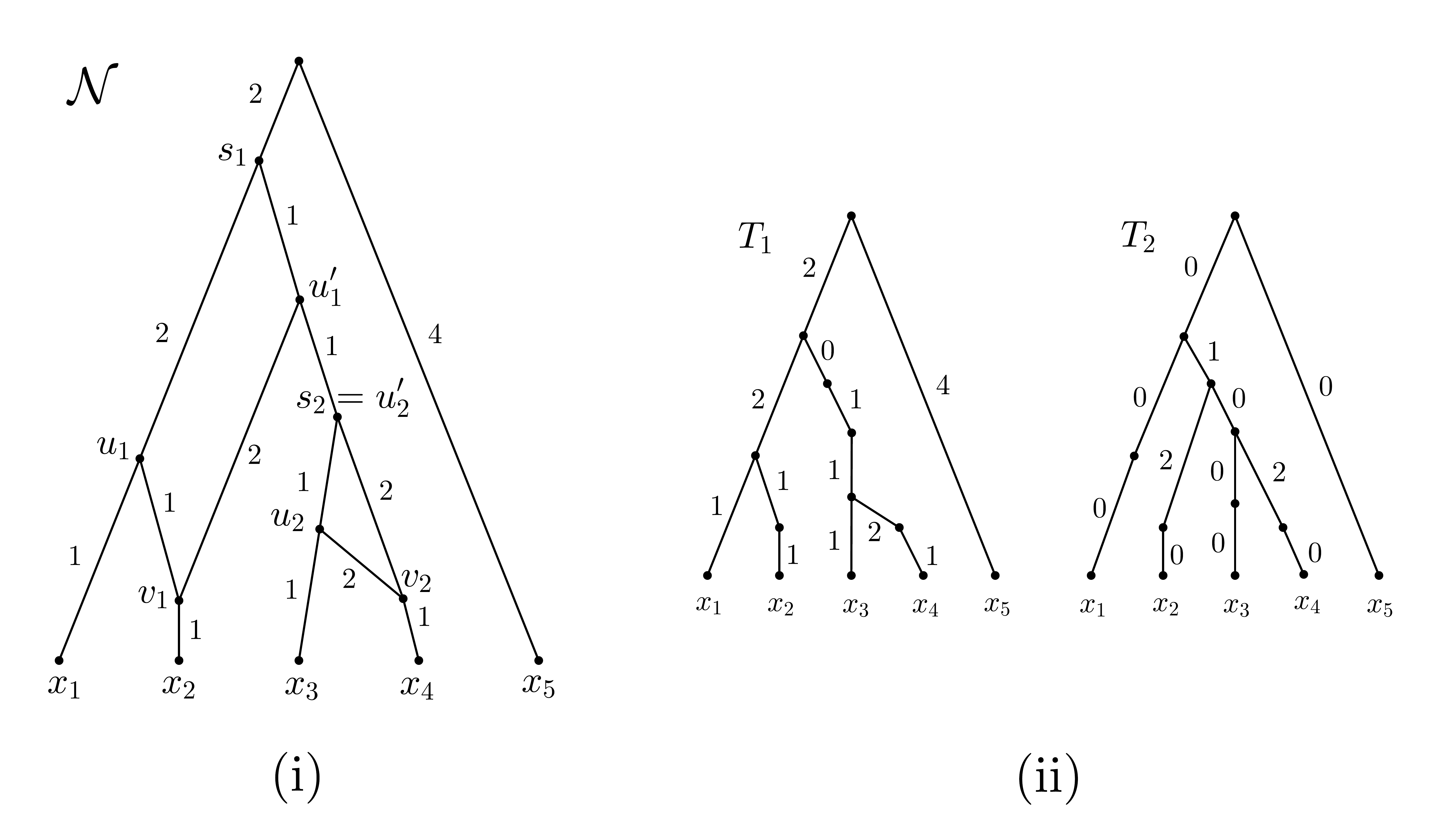}
    \caption{(i) A level-$1$ network $\cN$ on $X=\{x_1, x_2, \ldots, x_5\}$, and (ii) a weighted covering {$(T_1, T_2)$} of $\cN$.}
    \label{Fig:Level1}
\end{figure}

The following proposition is sufficient to show that {\sc Max-AllPaths PD} is polynomial time for the class of level-$1$ networks. The reason for this sufficiency is given after its proof.

\begin{proposition}
Let $\cN$ be a level-$1$ network on $X$, and let $(T_1,T_2)$ be a weighted covering of $\cN$. If $S\subseteq X$, then AllPaths-$PD_{\cN}(S)$ equates to the sum of $PD_{T_1}(S)$ and $PD_{T_2}(S)$.
\label{weightsum}
\end{proposition}

\begin{proof}
Let $A^{>0}({T_1})$ and $A^{>0}({T_2})$ denote the arcs of ${T_1}$ and ${T_2}$, respectively, of non-zero weight. For each $i\in \{1, 2\}$, let $\phi_i$ denote the identity map from $A^{>0}(T_i)$ to the arc set of $\cN$. By construction, for each $i$, the map $\phi_i$ is one-to-one and, for each non-zero weighted arc $e$ of $\cN$, the arc $e$ is in the co-domain of exactly one of $\phi_1$ and $\phi_2$. Now suppose that $S\subseteq X$, and let $\ell\in S$ and let $e$ be a non-zero weighted arc of $\cN$. Then $e$ is on a directed path from the root of $\cN$ to $\ell$ if and only if there is a unique $i$ such that $\phi^{-1}_i(e)$ has non-zero weight and $\phi^{-1}_i(e)$ is on the (unique) path in ${T_i}$ from its root to $\ell$. It now follows that
$$\mbox{AllPaths-PD}_{\cN}(S)={\rm PD}_{{T_1}}(S)+{\rm PD}_{{T_2}}(S).$$
\end{proof}

{As an example of Proposition~\ref{weightsum}, consider Fig.~\ref{Fig:Level1} and choose $S=\{x_2, x_4\}$. Then} AllPaths-PD$_{\cN}(S)=16$, PD$_{{T_1}}(S)=11$, and PD$_{{T_2}}(S)=5$. In particular,
$$\textrm{AllPaths-PD}_{\cN}(S)= {\rm PD}_{{T_1}}(S) + {\rm PD}_{{T_2}}(S).$$

Let $\cN$ be a level-$1$ network. It is clear that a weighted covering {$(T_1, T_2)$} of $\cN$ can be constructed in time polynomial in the size of {$\cN$}. {However, as $\cN$ is tree-child, the number of vertices, and thus arcs, in $\cN$ is linear in the size of $X$~\cite{Cardona2009, mcd15}, and so this construction can be done in time polynomial in the size of $X$.} By Proposition~\ref{weightsum}, finding the maximum value of AllPaths-PD$_{\cN}(S)$ over all subsets $S$ of $X$ of size $k$ is equivalent to finding the maximum value of ${\rm PD}_{{T_1}}(S)+{\rm PD}_{{T_2}}(S)$ over all subsets $S$ of $X$ of size $k$. The latter problem, called \textsc{Weighted Average PD on $2$ Trees}, is shown to be solvable in {time} polynomial {in the size of $X$} in~\cite{bor09a} {by reformulating the problem as a set of $|X|$ minimum-cost flow problems}. It follows that {\sc Max-AllPaths-PD} for the class of level-$1$ networks is also polynomial time {in the size of $X$}. {In particular, we have the following corollary.} 

\begin{corollary}
{Let $\cN$ be a level-$1$ network on $X$, and let $k$ be a positive integer. Then {\sc Max-AllPaths-PD}$(\cN, k)$ can be solved in time polynomial in the size of $X$.}
\end{corollary}

\section{Network-PD}
\label{sec:network-PD}

In this section, we turn to Network-PD, our variant of PD for phylogenetic networks that {is potentially} more realistic than AllPaths-PD discussed previously, but that requires additional information in the form of inheritance proportions on each reticulation arc. 
Let $\cN$ be a phylogenetic network on $X$ with an additional weight $p(e)$, the inheritance proportion, on each incoming arc to each reticulation. {This additional weight indicates} the proportion of features of the parent vertex present in the child vertex of that arc. Recall that, for any subset $S \subseteq X$, we defined Network-PD$_{\cN,\, p}(S)$ as 
$$\textrm{Network-PD}_{\cN,\, p}(S) = \sum_{e\in \cN}\gamma(S,e)\cdot w(e),$$
where, for each arc $e=(u,v) \in \cN$, the coefficient $\gamma(S,e)$ denotes the proportion of the features of $v$ that are present in the taxa set $S$. Thus, while AllPaths-PD implicitly assumes that each reticulation inherits all features present in both its parents, Network-PD allows us to model the fact that a reticulation representing, {for example}, a true hybridisation event might inherit only a proportion of features from each of its two parents.

In the following corollary of Theorems~\ref{thm:np-hard} and~\ref{thm:normal}, we observe that Network-PD is a generalisation of AllPaths-PD, from which it follows that maximising Network-PD is NP-hard. 

\begin{corollary}
The problem \textsc{Max-Network-PD} is NP-hard even when the inputted phylogenetic network is restricted to be from the class of normal networks. Moreover, \textsc{Max-Network-PD} cannot be approximated in polynomial time with approximation ratio better than $1-\tfrac{1}{e}$ unless ${\rm P}={\rm NP}$.
\end{corollary}

\begin{proof}
Given an input $(\cN,k)$ for \textsc{Max-AllPaths-PD}, we define an instance of \textsc{Max-Network-PD} as $(\cN,p,k)$ where $p(e)=1$ for all reticulation arcs. Since, for all {subsets} $S\subseteq X$, $\textrm{AllPaths-PD}_{\cN}(S)=\textrm{Network-PD}_{\cN,\, p}(S)$, both problems have the same optimal solution. As \textsc{Max-AllPaths-PD} is NP-hard, it follows that \textsc{Max-Network-PD} is also NP-hard and, {by Theorem~\ref{thm:np-hard}}, cannot be approximated in polynomial time with approximation ratio better than $1-\tfrac{1}{e}$ unless ${\rm P}={\rm NP}$.
\end{proof}

For a fixed arc $e$ and a subset $A$ of $X$, let  $\mathcal E_A$ be the event that a feature arising on arc $e$ is inherited by some taxon in set $A$. Then $\gamma(A,e)=\Pr[\mathcal E_A]$. For two subsets $A,B\subseteq X$, by the inclusion-exclusion principle
$$\gamma(A\cup B,e)=\Pr[\mathcal E_A\cup \mathcal E_B]=\Pr[\mathcal E_A]+\Pr[\mathcal E_B]-\Pr[\mathcal E_A\cap\mathcal E_B]\leq \gamma(A,e)+\gamma(B,e)-\gamma(A\cap B,e),$$
where the inequality is because $\mathcal E_{A\cap B}$ is a sub-event of $(\mathcal E_A\cap \mathcal E_B)$. Thus $\gamma$ is submodular. As for AllPaths-PD, we therefore obtain the following immediate corollary.

\begin{corollary}
A greedy algorithm returns a $1-\tfrac{1}{e}$ approximation for \textsc{Max-Network-PD}. Moreover, this approximation ratio cannot be improved unless ${\rm P}\neq {\rm NP}$.
\end{corollary}

\noindent {\bf {Remark.}} We have already seen that the case of \textsc{Max-Network-PD} in which $p(e)=1$ on all reticulation arcs is equivalent to \textsc{Max-AllPaths-PD} and is NP-hard. It is also the case that if $p(e)=0.5$ on all reticulation arcs (which might correspond to each reticulation being a perfect hybrid), then \textsc{Max-AllPaths-PD} is NP-hard. This can be seen by a reduction from {the NP-complete problem} \textsc{Exact-Cover-By-4-Sets}, {which takes as input a set $E$ with $|E|=4q$, and a collection $C$ of $4$-element subsets of $E$ with no element occurring in more than four subsets. The objective is to decide if there is a subset $C'$ of $C$ which is a partition of $E$. Such a subset is called an {\em exact cover}.} The construction is similar to that shown in Fig.~\ref{fig:hardness}, {where the leaf set is $C'$}, and there is an exact cover if and only if the optimal Network-PD of a subset of $|E|/4$ leaves is $|E|/4$ (each {leaf} contribut{ing} {four} weight-$1$ arcs but with $\gamma$ {averaging} 0.25).

\subsection{Max and Min Network-PD} \label{subsec:maxmin_network_pd}

In the previous section, we have seen that Network-PD is a generalisation of AllPaths-PD obtained by setting the inheritance proportion $p(e)$ to one for each reticulation arc $e$ of {a phylogenetic network} $\cN$. If we restrict the inheritance proportions such that, for a reticulation with incoming arcs $e_1$ and $e_2$ {to}
$$p(e_1)+p(e_2)=1,$$
we obtain the following relationship between the maximum (respectively, minimum) value of Network-PD and MaxWeightTree-PD (respectively, MinWeightTree-PD) as defined in Section~\ref{sec:PDvariants}.

\begin{theorem}
\label{thm:networkPDMaxMin}
Let $\cN$ be a phylogenetic network on $X$, and let $S$ be a fixed subset of $k$ elements of $X$. Let $R$ be the set of reticulation arcs of $\cN$, and let $P$ be the set of all functions mapping $R$ to $[0,1]^{|R|}$ with the restriction that at each reticulation the incoming arcs, $e_1,e_2$ say, have inheritance proportion{s} adding up to $1$, that is $p(e_1)+p(e_2)=1$. Then 
$$\max_{p\in P}\textup{Network-PD}_{\cN, \, p}(S)=\textup{MaxWeightTree-PD}_{\cN}(S),$$
and
$$\min_{p\in P}\textup{Network-PD}_{\cN, \, p}(S)=\textup{MinWeightTree-PD}_{\cN}(S).$$
\end{theorem}

\begin{proof}
{We prove the maximisation part of the theorem. The proof of the minimisation part is similar and omitted.}
Let $T$ be a connecting subtree for $S$ in $\cN$ such that $\textrm{PD}_{T}(S)=\textrm{MaxWeightTree-PD}_{\cN}(S)$.
Let $p_{\max}$ be {a function on $R$} defined {as follows}: at each reticulation {$v$} of $\cN$ that is in ${T}$, set $p_{\max}(e)=1$ if {$e$ is directed into $v$ and $e$ is in $T$} and {set} $p_{\max}(e)=0$ if {$e$ is directed into $v$ and $e$ is not in $T$}, and at each reticulation of $\cN$ that is not in ${T}$, set $p_{\max}$ to be $1$ on one incoming arc and $0$ on the other arbitrarily. Then, under $p_{\max}$, all features are inherited along the arcs of ${T}$ and hence 
$$\max_{p\in P}\textrm{Network-PD}_{{\cN},\, p}(S)\geq \textrm{Network-PD}_{p_{\max}}(S)= \textrm{MaxWeightTree-PD}(S).$$

Now consider $p_0\in P$ which maximises ${\rm Network-PD}_{{\cN},\, p_0}(S)$ and among all $p$ which maximise $\textrm{Network-PD}_{{\cN},\, p}(S)$, choose $p_0$ to have a minimum number of arcs $e$ with fractional inheritance proportion, i.e. such that $p_0(e)\not\in\{0,1\}$. Suppose that for all reticulation arcs $e$, {we have} $p_0(e)\in\{0,1\}$. Then at every reticulation one incoming arc has $p_0$ equal to $1$ and the other has $p_0$ equal to $0$. {Therefore} $\textrm{Network-PD}_{{\cN,\, p}}(S)$ is the PD of the minimal connecting tree of $S$ in the tree obtained from $\cN$ by deleting all {reticulation} arcs of $\cN$ with $p_0$ equal to $0$, and so
$$\max_{p\in P}\textrm{Network-PD}_{{\cN},\, p}(S)= \textrm{Network-PD}_{{\cN},\, p_{0}}(S)\leq \textrm{MaxWeightTree-PD}_{\cN}(S).$$
This proves the {maximisation} part of the theorem unless there is some arc with $p_0\not\in\{0,1\}$.

Now suppose, for {a} contradiction, that there is some {reticulation} arc with $p_0\not\in\{0,1\}$. Then there must be a reticulation $v$ with parents $u_1,u_2$ such that (i) $p_0(u_1),p_0(u_2)\not\in\{0,1\}$, (ii) for all reticulations that are descendants of $v$ the incoming arcs have $p_0$ equal to 0 or 1, and (iii) there is a subset $S'\subseteq S$ such that there is a path from $v$ to each element of $S'$ consisting of tree arcs and reticulation arcs with $p_0=1$. (Otherwise we can follow arcs down towards leaves until we find a last reticulation with inheritance from both parents, and if it is not an ancestor of any leaf in $S$, then we can {re}set its incoming arcs weights without affecting Network-PD($S$), {thereby reducing} the number of reticulations with {positive} inheritance from both parents.) 

Now consider the contribution to $\textrm{Network-PD}_{{\cN},\, p_0}(S)$ of each arc $e$ in $\cN$. This is $w(e)$ times the probability that a feature that arises on $e$ is inherited by some member of $S$, where the only randomness comes on reticulations which have $p_0\not\in\{0,1\}$. For {all subsets} $A\subseteq S$, let $\cE_A^e$ be the event that a feature that arises on $e$ is inherited by an element of $A$. Then the contribution of $e$ to $\textrm{Network-PD}_{{\cN},\, p_0}(S)$ is
$$w(e)\Pr[\cE_{S}^e]=w(e)\Pr[\cE_{S\setminus S'}^e]+w(e)\Pr[\overline{\cE_{S\setminus S'}^e}\cap\cE_{S'}^e],$$
where, {as above}, $S'$ is the subset of $S$ consisting of those elements in $S$ that can be reached from $v$ by paths whose reticulation arcs have $p_0=1$. Since all reticulation arcs below $v$ have {$p_0\in \{0, 1\}$}, {it follows that} if $e$ is a descendent arc of $v$, then $\Pr[\overline{\cE_{S\setminus S'}^e}\cap\cE_{S'}^e]=\Pr[\cE_{S'}^e]$ and {it} is either $0$ or $1$ but, importantly, {it} is independent of $p_0((u_1,v))$ {and} $p_0((u_2,v))$.

Let $\cE_{u_1}^e$ be the event that a feature that arises on $e$ is inherited down to the vertex $u_1$, and $\cE_{u_2}^e$ be the event that a feature that arises on $e$ is inherited down to the vertex $u_2$. If $e$ is not a descendant arc of $v$, then writing $q_0=p_0((u_1,v))$, and so $(1-q_0)=p_0((u_2,v))$, {we have}
$$\Pr[\overline{\cE_{S\setminus S'}^e}\cap\cE_{S'}^e]=\Pr[\overline{\cE_{S\setminus S'}^e}\cap\cE_{u_1}^e]q_0+\Pr[\overline{\cE_{S\setminus S'}^e}\cap\cE_{u_2}^e](1-q_0)-\Pr[\overline{\cE_{S\setminus S'}^e}\cap\cE_{u_1}^e\cap \cE_{u_2}^e]q_0(1-q_0).$$
So the full contribution from all arcs not descendants of $v$ is
$$\sum_e w(e) {\bigg(}\Pr[\cE_{S\setminus S'}^e]+\Pr[\overline{\cE_{S\setminus S'}^e}\cap\cE_{u_1}^e]q_0+\Pr[\overline{\cE_{S\setminus S'}^e}\cap\cE_{u_2}^e](1-q_0)-\Pr[\overline{\cE_{S\setminus S'}^e}\cap\cE_{u_1}^e\cap \cE_{u_2}^e]q_0(1-q_0) {\bigg)}.$$
{For convenience, write} $A=\sum_e w(e)\Pr[\cE_{S\setminus S'}^e]$, $B=\sum_e w(e)\Pr[\overline{\cE_{S\setminus S'}^e}\cap\cE_{u_1}^e]$, $C=\sum_e w(e)\Pr[\overline{\cE_{S\setminus S'}^e}\cap\cE_{u_2}^e]$, and $D=\sum_e {w(e)} \Pr[\overline{\cE_{S\setminus S'}^e}\cap\cE_{u_1}^e\cap \cE_{u_2}^e]$. {Note} that $D\leq\min\{B,C\}$. {Without loss of generality, we may assume} that $B\geq C$. {Then} we get the contribution to $\textrm{Network-PD}_{{\cN},\, p_0}(S)$ from all arcs not descendants of $v$ is
$$A +Bq_0 +C(1-q_0)-Dq_0(1-q_0)<A+B.$$
Hence if we amend $p_0$ by setting $p_0((u_1,v))=q_0=1$ {and} $p_0((u_2,v))=0$, then we can only be increasing $\textrm{Network-PD}_{{\cN},\, p_0}(S)$ and, {simultaneously}, reducing the number of arcs $e$ such that $p_0(e)\not\in\{0,1\}$. This contradicts the choice of $p_0$, and hence it must be that there is no arc with $p_0\not\in\{0,1\}$. {This completes the proof} of the theorem.
\end{proof}

Theorem~\ref{thm:networkPDMaxMin} gives us the following immediate corollary.
\begin{corollary}
{Let $\cN$ be a phylogenetic network on $X$, let $k$ be a positive integer, and let $P$ be the set of all inheritance proportion functions whereby, if $p\in P$, and $e_1$ and $e_2$ are the reticulation arcs directed into a reticulation of $\cN$, then $p(e_1)+p(e_2)=1$. Then} 
\[
 \max_{p\in P}\textsc{Max-Network-PD}(\cN, {p}, k) =  \textsc{Max-MaxWeightTree-PD}(\cN,k)
\]
and
\[
 \min_{p\in P}\textsc{Max-Network-PD}(\cN, {p}, k) =  \textsc{Max-MinWeightTree-PD}(\cN,k).
\]
\end{corollary}

\section{MaxWeightTree-PD}
\label{sec:max}

Given the importance of MaxWeightTree-PD and MinWeightTree-PD as bounds for Network-PD, we now analyse the complexity of determining the maximum possible PD score over all subsets of taxa of size $k$ under these two variants more in-depth. We begin by considering MaxWeightTree-PD and turn to MinWeightTree-PD in Section~\ref{sec:min}.

Let $\cN$ be a phylogenetic network on $X$. Recall that, for any subset $S\subseteq X$, we have defined MaxWeightTree-PD($S$) to be the maximum  of $\sum_{e \in T}w(e)$ over all connecting subtrees for $S$ in $\cN$.  
We now show that {the} corresponding optimisation problem \textsc{Max-MaxWeightTree-PD} {can} be solved in polynomial time by reducing \textsc{Max-MaxWeightTree-PD} to a minimum-cost flow problem, following the approach of \cite{bor09a}.

\begin{theorem}
Let $\cN$ be a phylogenetic network {on $X$}, and let $k$ be a positive integer. Then \textsc{Max-MaxWeightTree-PD} applied to $\cN$ and $k$ can be solved in polynomial time.
\end{theorem}

\begin{proof}
Starting with $\cN$, we define a flow network by
\begin{itemize}%
\setlength{\itemsep}{0pt plus 0pt}%
\item setting the root $\rho$ of $\cN$ to be the source,
\item adding additional arcs from $\rho$ to each tree vertex of $\cN$, which we shall call the \textit{extra arcs},
\item appending a new vertex $t$ with an arc from each leaf of $\cN$ directed into $t$, and a {new vertex} $t'$, the sink, with a single arc from $t$ to $t'$,
\item setting the capacity of all arcs inherited from $\cN$, and the arcs from the leaves to $t$, to be 1,
\item setting the capacity of the extra arcs and the final arc $(t,t')$ to be $k$, {and}
\item setting the cost of each arc $e$ inherited from $\cN$ to be the negative of its weight, {that is} $-w(e)$, and the cost of all the additional arcs to be $0$.
\end{itemize}
Observe that, due to the arc $(t, t')$, the maximum flow from the source $\rho$ to the sink $t'$ is $k$ units. Therefore, as all capacities are integral, there is a minimum-cost integral flow of $k$ units that may be found in polynomial time (see, for example,~\cite{ahuja1993network, bor09a}). 

Since there is a cut of the flow network between the leaves in $\cL(\cN)$ and $t$, and each of these arcs has capacity 1, exactly $k$ of these arcs are used in the minimum-cost integral flow. We denote the set of leaves adjacent to these arcs by $S$. Note that if a minimum-cost flow has non-zero flow through any extra arc $(\rho,v)$, then there is a minimum-cost flow that has non-zero flow in the arc $(u,v)$ of $\cN$ directed into $v$, since there is a path from $\rho$ to $v$ via $(u,v)$ which has lower cost than the path {from $\rho$ to $v$ via the extra arc $(\rho, v)$} and is not at capacity due to the extra arcs.  (In the case that the arc $(u,v)$ has weight zero, this path actually has the same cost but, without loss of generality, we can still assume our minimum-cost flow routes through $(u,v)$.) 

Therefore the set of arcs of $\cN$ that have non-zero flow form a connecting subtree $T_S$ for $S$, since (i), by the argument of the previous paragraph, there must be flow from the root to each leaf in $S$ and (ii), each arc {directed out of} a reticulation has capacity~$1$, so {at most} one incoming arc to a reticulation has non-zero flow. The total cost of such a flow is exactly the negative of the sum of the weight of arcs in the {$T_S$}. Moreover, any connecting subtree $T_S$ for $S$ can be realised as a flow of $k$ units by routing as much flow as possible through the arcs constituting $T_S$, and routing extra flow through the extra arcs as necessary. Since we can find the minimum-cost integral $k$-flow in polynomial time, we can therefore find the max-weight embedded connecting subtree for any set of $k$ leaves of $\cN$ in polynomial time.
\end{proof}

\noindent {\bf {Remark.}} The proof of the last theorem can be easily extended to show that the MaxWeightTree-PD can be optimised {for the following problem} using a construction {similar} to that used in~\cite{bor09a} {for the analogous optimisation problem for phylogenetic trees}. 

\noindent\textbf{\textsc{Weighted Average PD on 2 Networks}} \\
\textbf{Input:} {Two phylogenetic networks $\cN_1$ and $\cN_2$ on $X$ with (arc) weight functions $w_1$ and $w_2$, and a positive integer $k$.} \\
\textbf{Objective:} {Determine the maximum value of}
{$$\sum_{e\in T_1} w_1(e) + \sum_{e\in T_2} w_2(e),$$}
{where $T_1\in \mathsf{T}_S(\cN_1)$ and $T_2\in \mathsf{T}_S(\cN_2)$, over all subsets $S\subseteq X$ of cardinality $k$.}

\section{MinWeightTree-PD}
\label{sec:min}

Recall that, for a phylogenetic network {on $X$}, we defined \textsc{Max-MinWeightTree-PD} to be the problem of determining the maximum weight, over {all} subsets of $X$ of cardinality $k$, of the minimum weight connecting subtree of the subset. Our first observation is that for an arbitrary phylogenetic network {$\cN$} on $X$, even computing $\textrm{MinWeightTree-PD}_{{\cN}}(X)$ is computationally hard. To make this more precise consider the following problem:

\noindent\textbf{\textsc{Minimum-Weight $X$-Tree}{$(\cN)$}} \\ 
\textbf{Input:} A phylogenetic network $\cN$ on taxa set $X$. \\
\textbf{Objective:} The value of MinWeightTree-PD$_{\cN}(X)$, i.e. the minimum weight of a connecting subtree for $X$ in $\cN$.

\noindent We will show that this problem is hard by making use of a reduction {from} the well-known NP-complete problem \textsc{Exact Cover By $3$-Sets}~\cite{kar72}:

\noindent\textbf{\textsc{Exact Cover By $3$-Sets}{$(X, C)$}} \\
\textbf{Input:} A set $X$ with $|X|=3q$, and a collection $C$ of $3$-element subsets of $X$ with no element occurring in more than three subsets. \\
\textbf{Objective:} Determine if $C$ contains an exact cover of $X$, that is a subset $C'$ of $C$ which is a partition of $X$?

\begin{theorem}
\label{Thm:MinWeightHard}
The problem \textsc{Minimum-Weight $X$-Tree} is NP-hard.
\end{theorem}

\begin{proof}
Take an instance of {\sc Exact Cover By $3$-Sets}:, i.e. a set $X$ with $|X|=3q$, and a collection $C=\{C_1, C_2, \ldots, C_k\}$ of $3$-element subsets of $X$ with no element occurring in more than three subsets. Similar to a construction in~\cite{kar72}, construct a phylogenetic network on $X$ as follows. Let $D$ be the rooted acyclic digraph with vertex set $C\cup X\cup \rho$ and arc set
$$\{(\rho, C_i): i\in \{1, 2, \ldots, k\}\}\cup \{(C_i, x): x\in X\cap C_i\}.$$
Now weight the arcs of $D$ so that $(\rho, C_i)$ has weight $3$ for all $i$ and all remaining arcs have weight zero. We next construct a phylogenetic network $\cN$ on $X$ from $D$ and its weighting. First, refine the vertices $\rho$ and $C_i$ for all $i$ so that all vertices (except elements of $X$) with in-degree zero or in-degree one have out-degree two. Second, for each element $x$ in $X$, adjoin a new vertex to it via a new arc so that the new vertex has in-degree one (and out-degree zero) and relabel so that the resulting new vertex is now $x$. Third, refine each vertex with in-degree three so that no vertex has in-degree more than two. Lastly, extend the weighting of $D$ by assigning all (new) unweighted arcs weight zero. The resulting phylogenetic network {on $X$} is $\cN$. To illustrate {the construction}, the top half of $D$ and its weighting, and a possible top half of $\cN$ is shown in Fig.~\ref{half}(i) and~(ii), respectively. Clearly, $\cN$ can be constructed in time polynomial in the size of the initial instance of {\sc Exact Cover By $3$-Sets}. Furthermore, it is easily seen that $C$ contains an exact cover of $X$ if and only if $\cN$ has a connecting subtree for $X$ of weight at most $|X|$. Hence computing MinWeightTree-PD$_{\cN}(X)$ is NP-hard. 
\end{proof}

\begin{figure}[htbp]
\centering
\includegraphics[scale=0.25]{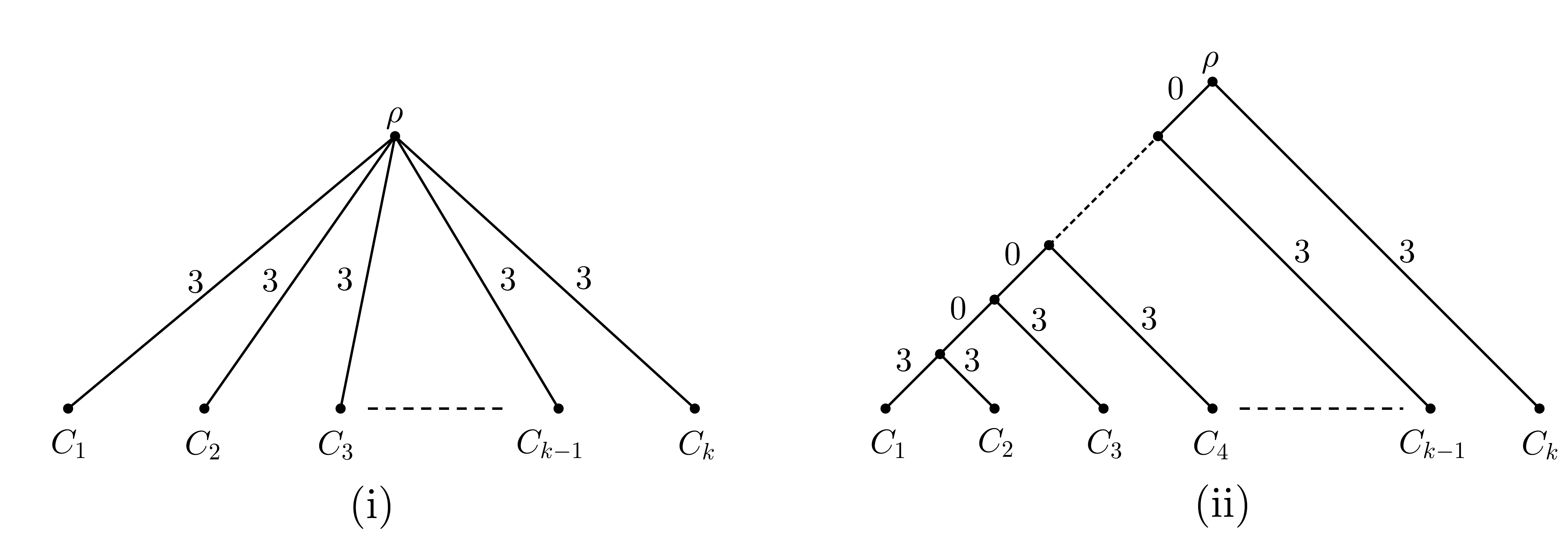}
\caption{(i) The top half of $D$ and its weighting, and (ii) a possible top half of $\cN$.}
\label{half}
\end{figure}

Although we have shown that computing $\textrm{MinWeightTree-PD}_{\cN}(X)$ is hard on a general phylogenetic network {$\cN$}, if we restrict $\cN$ to be a tree-child network, the problem of computing $\textrm{MinWeightTree-PD}_{\cN}(X)$ can be solved in polynomial time.

\begin{theorem}
\label{Thm:MinWeightTC}
Let $\cN$ be a tree-child network on $X$. Then \textsc{Minimum-Weight $X$-Tree} applied to $\cN$ can be solved in polynomial time in the size of $X$.
\end{theorem}

\begin{proof}
Let {$T$ be a connecting subtree for $X$ in $\cN$.} Since $\cN$ is tree-child, $T$ contains every tree arc of $\cN$ and $T$ contains, for each reticulation $v$ of $\cN$, precisely one reticulation arc directed into $v$~\cite{Semple2015}. Thus, to find a minimum-weight connecting subtree for $X$ in $\cN$ it suffices to determine for each reticulation $v_i$, a reticulation arc of minimum weight directed into $v_i$. In particular, if $e_i$ is such an arc for all $i$, then the arc set of a minimum-weight connecting subtree for $X$ in $\cN$
is the union of the set of tree arcs in $\cN$ and $\{e_i: v_i \text{ is a reticulation in } \cN\}$. This completes the proof {of the theorem}.
\end{proof}

In contrast to the last theorem, computing MinWeightTree-PD$_{\cN}(S)$ for a \emph{given} subset $S \subset X$ of a phylogenetic network $\cN$ on $X$ is hard even if $\cN$ is a normal network (and so, in particular, if $\cN$ is tree-child).

\begin{theorem}
\label{Thm:MinWeightPDSubset}
Let $\cN$ be a phylogenetic network on $X$ and let $S \subset X$ be a strict subset of {$X$}. Then, computing $\textrm{MinWeightTree-PD}_{\cN}(S)$ is {\rm NP}-hard even if $\cN$ is a normal network.
\end{theorem}

\begin{proof}
Take an instance of \textsc{Minimum-Weight $X$-Tree}, i.e. an arbitrary phylogenetic network $\cN$ on $X$. We construct a normal network $\cN'$ on $X' \supset X$ from $\cN$ by subdividing all reticulation arcs and {adjoining} a {new} leaf {via a new arc to} each new vertex. If $(u,v)$ was a reticulation arc in $\cN$ with weight $w(u,v)$, we assign weight $w(u,v)$ to the arc of the subdivision directed into the corresponding reticulation in $\cN'$, and we assign weight zero to the other arc of the subdivision {as well as to the incident} pendant arc leading to {a new leaf} in $X'-X$. {Setting} $S=X \subset X'$, the problem of calculating $\textrm{MinWeightTree-PD}_{\cN'}(S)$ for the normal network $\cN'$ on $X'$ corresponds to the problem of calculating $\textrm{MinWeightTree-PD}_{\cN}(X)$ for the arbitrary network $\cN$ on $X$. However, {by} Theorem~\ref{Thm:MinWeightHard}, the latter {problem} is NP-hard, and hence computing $\textrm{MinWeightTree-PD}_{\cN'}(S)$ for the normal network $\cN'$ on $X'$ and $S \subset X'$ is NP-hard. This completes the proof {of the theorem}.
\end{proof}

We now turn to the original problem of this section and show that it is again an NP-hard problem.

\begin{theorem}\label{Thm:MaxMinWeightPD}
The problem \textsc{Max-MinWeightTree-PD} is NP-hard.
\end{theorem}

\begin{proof}
Take an instance of \textsc{Minimum-Weight $X$-Tree}, i.e. a phylogenetic network $\cN$ on $X$ with $|X|=k$. We now construct a phylogenetic network $\cN'$ on $X' \supset X$ as follows:
\begin{itemize}\setlength{\itemsep}{0pt plus 0pt}
    \item Choose a pendant arc $e$ leading to a leaf, say $x$, of $\cN$, subdivide it (possibly several times), and {adjoin a new leaf via a new arc with weight zero to each new vertex}.
    \item If the weight of $e$ in $\cN$ was $w(e)$, assign weight $w(e)$ to the arc incident with $x$, and assign weight zero to all other arcs of the subdivision.
\end{itemize}
Now, consider the instance $(\cN', k)$ of the problem \textsc{Max-MinWeightTree-PD}, i.e. consider the problem of computing the maximum value of $\textrm{MinWeightTree-PD}_{\cN'}(S)$ over all subsets $S \subseteq X'$ of cardinality $|X|=k$ on $\cN'$. As all elements in $X'{-}X$ {are} attached to $\cN'$ {via} pendant arcs of {weight} zero and all {non-pendant} arcs on a path from the root of some connecting subtree $T$ for $X'$ in $\cN$ to elements in $X'-X$ are also covered by a path from the root of $T$ to $x$, {there is} a subset $S \subseteq X'$ of cardinality $k=|X|$ maximising $\textrm{MinWeightTree-PD}_{\cN'}(S')$ over all subsets $S' \subseteq X'$ with $|S'|=|X|$ {that} does not contain any of the leaves {in $X'-X$}. In other words, we can assume that $S = X$. Thus, the maximum value of MinWeightTree-PD$_{\cN'}(S')$ in $\cN'$ over all subsets $S' \subseteq X'$ with $|S'|=|X|$ coincides with the value of MinWeightTree-PD$_{\cN}(X)$ in $\cN$. {By} Theorem \ref{Thm:MinWeightHard}, the latter {problem} is NP-hard, {and so} we conclude that the problem \textsc{Max-MinWeightTree-PD} is {also} NP-hard.
\end{proof}

\section{Concluding remarks}

Phylogenetic diversity is widely used for quantifying the biodiversity of a set of species based on their evolutionary history and relatedness. Traditionally, PD was calculated on a phylogenetic tree representing the evolution of a set of species. However, it is now commonly accepted that evolution is not always tree-like and that many species' evolutionary history contains reticulation events such as hybridisation or lateral gene transfer. In this paper, we therefore defined four natural variants of the PD score for a subset of taxa whose evolutionary history is represented by a phylogenetic network. Under these variants, we considered the computational complexity of, given a positive integer $k$, determining the maximum PD score over all subsets of taxa of size $k$ when the input is restricted to different classes of phylogenetic networks. More precisely, we showed that determining the maximum PD score over all subsets of taxa of size $k$ under AllPaths-PD is NP-hard even when the inputted phylogenetic network is restricted to be from the class of normal networks. However, the problem is solvable in polynomial time for the class of level-1 networks. The corresponding maximisation problem is also NP-hard under Network-PD and MinWeightTree-PD (again, even when the inputted phylogenetic network is restricted to be from the class of normal networks), but it is solvable in polynomial time under MaxWeightTree-PD. 

An interesting question, however, is to determine the computational complexity of the problems \textsc{Max-Network-PD} and \textsc{Max-MinWeightTree-PD} when the inputted phylogenetic network is restricted to be from the class of level-1 networks. We leave this problem to future research. 

\section{Acknowledgements}
All authors thank Schloss Dagstuhl---Leibniz Centre for Informatics---for hosting the Seminar 19443 \emph{Algorithms and Complexity in Phylogenetics} in October 2019, where this work was initiated and Prof.~Mike Steel for hosting a workshop in Sumner, New Zealand in March 2020. The first and second authors thank the New Zealand Marsden Fund for financial support, and the third author thanks The Ohio State University President’s Postdoctoral Scholars Program. The first author thanks the Erskine Visiting Fellowship Programme for supporting their extended visit to the University of Canterbury, New Zealand in 2020.

\end{document}